\documentclass[twocolumn,aps,showpacs,eqsecnum,10pt]{revtex4-2} 
\usepackage{latexsym}
\usepackage{amsthm}
\usepackage{amsmath}
\usepackage{amssymb}
\usepackage{epsfig}
\usepackage{graphicx}
\usepackage{dcolumn}
\usepackage{subfigure}
\usepackage{bm}
\usepackage{color,soul}
\usepackage{xcolor}
\usepackage{hyperref}

\newtheorem{theorem}{Theorem}

\newtheorem{example}{Example}

\newcommand{\ket}[1]{|#1\rangle}
\newcommand{\bra}[1]{\langle #1|}
\newcommand{\tr}[1]{{\rm tr}[#1]}

\begin{document}
\title{Minimum-error discrimination of thermal states}
\author{Seyed Arash Ghoreishi}
\email{arash.ghoreishi@savba.sk}
\author{Mario Ziman}
\affiliation{RCQI, Institute of Physics, Slovak Academy of Sciences, D\'ubravsk\'a cesta 9, 84511 Bratislava, Slovakia}
\begin{abstract}
We study several variations of the question of minimum-error discrimination of thermal states. Besides of providing the optimal values for the probability error we also characterize the optimal measurements. For the case of a fixed Hamiltonian we show that for a general discrimination problem the optimal measurement is the measurement in the energy basis of the Hamiltonian. We identify a critical temperature determining whether the given temperature is best distinguishable from thermal state of very high, or very low temperatures. Further, we investigate the decision problem of whether the thermal state is above, or below some threshold value of the temperature. Also, in this case, the minimum-error measurement is the measurement in the energy basis. This is no longer the case once the thermal states to be discriminated have different Hamiltonians. We analyze a specific situation when the temperature is fixed, but the Hamiltonians are different. For the considered case, we show the optimal measurement is independent of the fixed temperature and also of the strength of the interaction.     
\end{abstract}
\maketitle

\section{Introduction}
Temperature $T$, taking positive values if measured in Kelvins,
is one of the fundamental physical parameters of systems in so-called
thermal equilibrium states with their surrounding
(playing the role of the heat bath) \cite{ChangOxford2004,PuglisiRep2017} . For quantum systems the temperature
is assigned to density operators of the form
\begin{equation}
\rho_{\beta}(H)=\frac{e^{-\beta H}}{{\rm tr}{e^{-\beta H}}}\,,
\end{equation}
where $\beta=\frac{1}{k_B T}$ is the so-called inverse temperature,
$H$ is the system's Hamiltonian and $k_B$ is the Boltzmann constant.
These states are named as thermal, or Gibbs states, and
$Z_\beta={\rm tr}e^{-\beta H}$ is the so-called partition function \cite{GemmerBook2004,VinjanampathyContemp2016,DeffnerBook2019,AndersNJP2013}.
For a given thermal state $\rho_\beta (H)$ the probability distribution
of the system's energy reads $w_\beta(E_j)=e^{-\beta E_j}/Z_\beta$.  

In this paper we address the questions related to the distinguishability
of temperatures. The questions of distinguishability of quantum states
represent one of the basic fundamental problems of quantum theory.
There are several mathematical variations of this question depending on our
particular purposes: one might be interested in how to minimize the error
of our conclusions \cite{HelstromBook1976}, or one could question
the ability to make error-free
(unambiguous) conclusions \cite{IvanovicPLA1987,PeresPLA1988,DieksPLA1988}, or some mixed strategies to achieve the balance between the mentioned methods \cite{CrokePRL2006}. In its simplest form
the task is the following: consider a source of
systems described either by the thermal state $\varrho_\beta(H)$, or
$\varrho_{\beta^\prime}(H^\prime)$. We are asked to perform an experiment
on a finite number of copies and conclude the identity of
the state. It is known that such a decision cannot be done perfectly unless
the supports of the density operators are mutually orthogonal. Since thermal
states of nonzero temperature are full-rank, it follows that only
mutually orthogonal ground states (zero temperature) of suitable
Hamiltonians can be discriminated perfectly.

As we already mentioned there are two conceptually different ways how to formulate imperfections in the decision problems. In both cases we optimize average probabilities of the performance. When we allow our conclusions might be errorneous, then our goal is to minimize (on average) the number of (random, but) invalid conclusions, i.e. the probability of error. However, if we do not allow conclusions to be incorrect, then we must allow also inconclusive outcomes. Our aim, in such a case, is to minimize the rate of inconclusive outcomes, that is, the probability of failure. Both these problems have been extensively studied in literature \cite{BergouLNP2004,BarnettAOP2009,BaeJPA2015,BarnettJPA2009,BanIJTP1997,BarnettPRA2001,ChouPRA2003,AnderssonPRA2002,ChouPRA2004,MochonPRA2006,BaeNJP2013,GhoreishiQIP2019}.  It is known that full-rank density operators cannot be unambiguously discriminated. Therefore, for a pair of thermal states the error-free conclusion is possible only if one of the states is the ground state. Moreover, if one of the states is of nonzero temperature, then only the identity of the nonzero temperature can be concluded in an error-free manner. Consequently, the questions that remain are related to minimum-error discrimination.

The problem of discrimination between thermal baths by quantum probes was investigated recently \cite{GiananiPRR2020,CandeloroPRA2021}. Based on the fact that nonequilibrium states of quantum systems in contact with thermal baths can be used to distinguish environments
with different temperatures, they study a more generic problem
that consists in discriminating between two baths with different temperatures and show that the existence of coherence in the initial state preparation is beneficial for
the discrimination capability  \cite{GiananiPRR2020}. Moreover the same problem was addressed by dephasing quantum probes and it has been shown that the dephasing quantum probes are useful in discriminating low values of the interaction time and a better performance can be obtained for intermediate values of interacting times \cite{CandeloroPRA2021}. 

Our goal is simple: evaluate the minimum error probability for thermal states and analyze the results. We are interested in understanding when the difference of the temperatures matters, i.e. increases the distingusihability, and how the parameters of the Hamiltonian affects the distinguishability. Is it easier to distinguish larger, or smaller temperatures? Does the strength of the interaction increases the distinguishability, or not? 

This article is organized as follows. In Sec. \ref{sec2} we study the question of discrimination of two states with different temperatures in the presence of same Hamiltonian and show that the measurement is optimal in the energy basis and as a case study we investigate the problem for qubit states and qudits with the same energy separations between energy levels. We generalize to the case of many states in Sec. \ref{sec3} and consider the problem of temperature threshold. In Sec. \ref{sec4} we study the problem of discrimination of thermal states with different Hamiltonians. Finally, in Sec. \ref{conclusion} we close the article with some concluding remarks.

\section{Binary case with fixed Hamiltonian} \label{sec2}
For a general binary minimum-error discrimination the optimal value of success probability is known \cite{HelstromBook1976} to be given by the Helstrom formula
\begin{equation}
  p_{\rm error}=\frac{1}{2}(1-\frac{1}{2}||\varrho_1-\varrho_2||_1),
\end{equation}
where $||\Delta||_1={\rm tr}|\Delta|$ is the trace norm (sum of square roots of eigenvalues of $\Delta\Delta^\dagger$). Our goal is to analyze this quantity
for a pair of thermal states of different temperatures of $d$-dimensional
quantum system described by a fixed Hamiltonian $H$. For simplicity,
we will denote the states
as $\varrho_j=\varrho_{\beta_j}(H)$ and label the eigenvalues of $H$
in an increasing order (including degeneracies) as follows
$E_0\leq E_1\leq\cdots\leq E_{d-1}$. Let us denote by $\Pi_j$ projectors
onto the eigenvectors associated with eigenvalues $E_j$, respectively.
Then $\varrho_1=(1/Z_1)\sum_j\exp(-\beta_1E_j) \Pi_j$,
$\varrho_2=(1/Z_2)\sum_j\exp(-\beta_2E_j) \Pi_j$ 
and for the trace distance we get
\begin{equation}
||\rho_1-\rho_2||_1= \sum_j\left|\frac{\exp(-\beta_1E_j)}{Z_1}-\frac{\exp(-\beta_2E_j)}{Z_2}\right|,
\end{equation}
where $Z_1=\sum_j \exp(-\beta_1 E_j)$ and $Z_2=\sum_j \exp(-\beta_2 E_j)$.
Using the associated energy distributions $w_{1j}=w_{\beta_1}(E_j)$
and $w_{2j}=w_{\beta_2}(E_j)$ the minimum error probability equals
\begin{equation}
  p_{\rm error}=\frac{1}{2}(1-\frac{1}{2}\sum_j |w_{1j}-w_{2j}|)\,.
\end{equation}

It follows from the general consideration of the minimum error discrimination
problems that the optimal measurement is formed by a pair of projectors
$Q_1$, $Q_2$ ($Q_1+Q_2=I$) indicating the temperatures $\beta_1,\beta_2$,
respectively. Let us note that for the conclusion associated with $Q_1$ the
outcome probabilities satisfy
the relation $\tr{\varrho_1 Q_1}>\tr{\varrho_2 Q_1}$
and similarly for $Q_2$ we have $\tr{\varrho_2 Q_2}>\tr{\varrho_1 Q_2}$.
In other words, for the considered measurement outcome the conclusion
of the discrimination results in the state maximizing the probability
of the considered outcome. In what follows we will investigate the same
``maximum likelihood" decision strategy in the case of a measurement
in the energy basis.

\subsection{Measurement in energy basis is optimal}
Suppose the considered minimum-error discrimination problem of two states
is going to be decided
in some $n$-valued measurement with effects $F_1,\dots,F_n$ ($F_j\geq O$, $\sum_j F_j=I$).
Let us denote by $f_{xj}=\eta_xw_{xj}$ (with $w_{xj}=\tr{\varrho_x F_j}$) the probability that the state
$\varrho_x$ has lead to the outcome $j$. The probability $\eta_x$ ($\eta_1+\eta_2=1$)
describes the a priori distribution with which the states are sampled. For simplicity 
we assume an unbiased case and set $\eta_x=1/2$, i.e. $q_{xj}=\frac{1}{2}w_{xj}$.
For each outcome $j$ either $q_{1j}>q_{2j}$, or $q_{1j}<q_{2j}$, or $q_{1j}=q_{2j}$.
In accordance with the maximum likelihood decision strategy, if $q_{1j}>q_{2j}$ we conclude
the state is $\varrho_1$. If $q_{1j}<q_{2j}$, then the observation of $j$ leads us to
the conclusion $\varrho_2$. If the probabilities coincide, $q_{1j}=q_{2j}$, then both
options are equal, thus, we toss an unbiased coin to decide whether $\varrho_1$,
or $\varrho_2$. The errors for unbalanced cases equal to the minimal values of $q_{1j}$ and
$q_{2j}$. The error for the balanced case equals to $q_{1j}=q_{2j}$. It follows that the whole
error probability can be written as $p_{\rm error}=\frac{1}{2}\sum_j \min\{w_{1j},w_{2j}\}$.
Further we will use the identity $\min\{w_{1j},w_{2j}\}=\frac{1}{2}(w_{1j}+w_{2j}-|w_{1j}-w_{2j}|)$.
Summing up over $j$ we end up with the formula
$p_{\rm error}=\frac{1}{2}(1-\frac{1}{2}\sum_j |w_{1j}-w_{2j}|)$. Let us note
  that the analogous approach applies also for the problem of discrimination between
  two coherent states by means of the photon number detector \cite{SychPRL2016}.

Let us now consider that the measurement performed is the energy measurement, i.e.
outcomes $E_j$ associated with projectors $\Pi_j$, hence,
$w_{xj}=\tr{\varrho_x\Pi_j}=(1/Z_x)\exp(-\beta_x E_j)$. Inserting these numbers into
the above formula we observe that the obtained expression coincides with
the optimal value for the minimum-error discrimination of two thermal states.
As a result it follows that the measurement of energy implements the optimal
discrimination measurement of thermal states. The observation of the outcome $E_j$
is associated with the thermal state, for which the probability of this outcome
is larger. If the probabilities for both options are the same, we toss a coin
to make the decision. Let us stress that although the probability distributions
$w_{xj}$ for the thermal states are decreasing as the energies are increasing, it is not
straightforward that observation of larger energies implies the conclusion for the state of
larger temperature.

\subsection{Ground-state discrimination}
Consider a special case when $T_1=0$ (i.e. $\beta_1\to\infty$), thus, we aim to distinguish
the ground state $\varrho_1=\frac{1}{\tr{\Pi_0}}\Pi_0$ from a thermal state $\varrho_2=(1/Z_2)\exp{-\beta_2 H}$. Let us remark that for the ground state the energy measurement
necessarily (and with certainty) leads to the outcome $E_0$. As a result,
the observation of any other energy implies the state is thermal and such conclusion
is error-free, thus, unambiguous. Therefore, the  minimum-error probability equals
$p_{\rm error}=\frac{1}{2}\tr{\varrho_2\Pi_0}=\frac{1}{2}(1/Z_2)e^{-\beta_2 E_0}<\frac{1}{2}$.
Since only the observation of $E_0$ is inconclusive in the unambiguous sense, it follows
the probability $p_{\rm failure}=\frac{1}{2}\tr{\Pi_0(\varrho_1+\varrho_2)}
=\frac{1}{2}(1+p_{\rm error})$ characterizes
the failure probability for the unambiguous discrimination between the ground state and
arbitrary thermal state.

\subsection{Case study: Qubit}

Let us investigate in more detail the simplest the case of the qubit system.
The general Hamiltonian has the form $H=aI+\vec{\alpha}\cdot\vec{\sigma}$.
Let us introduce the quantity $\alpha=||\vec{\alpha}||>0$ and operator
$S=\vec{\alpha}\cdot\vec{\sigma}=\ket{\varphi_+}\bra{\varphi_+}-\ket{\varphi_-}\bra{\varphi_-}$,
where $\ket{\varphi_\pm}$ are eigenvectors of $H$ and $S$.
The eigenvalues of $H=aI+\alpha S$ reads $E_\pm=a\pm \alpha$, thus,
we obtain energies $E_0=E_-$ for the ground state and $E_1=E_+$ for the excited one.
Direct calculation gives
\begin{equation}
p_{\rm error}=\frac{1}{2}\bigl(1-\frac{1}{2}\frac{|\sinh{\alpha}(\beta_1-\beta_2)|}{\cosh(\beta_1\alpha)\cosh(\beta_2\alpha)}\bigr)\, .
\end{equation}
The optimal measurement consists of projectors $\Pi_0=\frac{1}{2}(I-S)=\ket{\varphi_-}\bra{\varphi_-}$
and $\Pi_1=\frac{1}{2}(I+S)=\ket{\varphi_+}\bra{\varphi_+}$.
Further, without loss of generality we may assume $\beta_1>\beta_2$ (i.e. $T_1<T_2$). Then
\begin{eqnarray}
  \nonumber
  w_{10}=\tr{\varrho_1 \Pi_0}=\frac{e^{\beta_1 \alpha}}{Z_1}  &>&\frac{e^{\beta_2 \alpha}}{Z_2}=\tr{\varrho_2\Pi_0}=w_{20} \\
  \nonumber
  w_{11}=\tr{\varrho_1 \Pi_1}=\frac{e^{-\beta_1 \alpha}}{Z_1}&<&\frac{e^{-\beta_2 \alpha}}{Z_2}=\tr{\varrho_2\Pi_1}=w_{21} \,,
\end{eqnarray}
where $Z_j=2\cosh{\beta_j\alpha}$. Therefore, if the ground energy is recorded we conclude
the smaller of the temperatures ($T_1$), whereas the if excited energy is observed, then the state of the larger temperature ($T_2$) is identified.

Let us note that the minimum error probability in the qubit case
is independent of the parameter $a$. In fact, this is just the reflection
of the fact that the system's energy is specified up to an additive constant.
In particular, suppose Hamiltonians, e.g., $H$ and $H^\prime=H+xI$ for arbitrary $x$. The direct calculation shows that the value of $x$ does not affect the thermal state 
$\frac{e^{-\beta H^\prime}}{\tr{e^{-\beta H^\prime}}}=
\frac{e^{\beta x}e^{-\beta H}}{\tr{e^{\beta x}e^{-\beta H}}}=
\frac{e^{-\beta H}}{\tr{e^{-\beta H}}}$. Without loss of generality
we may set the ground state energy to zero ($E_0=0$), or, if it suits us,
we may assume the Hamiltonian is traceless. 
Figure \ref{3d} shows the three-dimensional (3D) plot of the probability of error in terms of two temperatures $T_1$ and $T_2$ in the presence of fixed Hamiltonian $H=\sigma_z$.  As it can be seen from the figure the probability of error is symmetric under the exchanging $T_1$ and $T_2$. 
\begin{figure}
  \mbox{\includegraphics[scale=1.6]{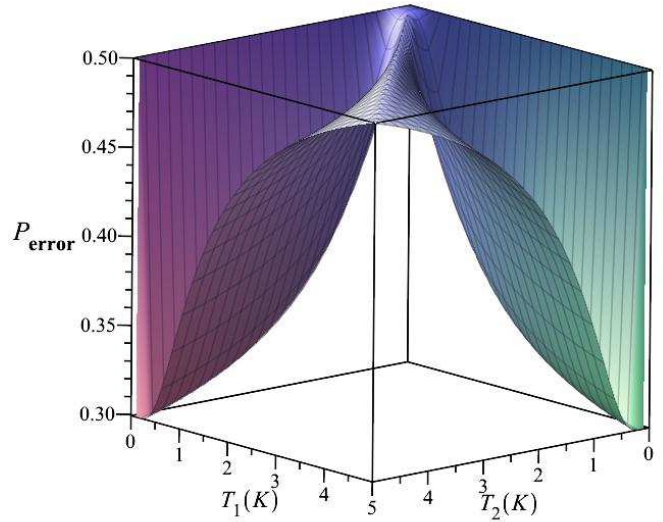}}
\caption{3D plot of probability of error between two satets $\varrho_1$ and $\varrho_2$  in terms of $T_1$ and $T_2$  in the presence of a fixed Hamiltonian $H=\sigma_z$. The probability of error is symmetric under the exchanging $T_1$ and $T_2$. }
\label{3d}
\end{figure}

For $T_1=0$ ($\beta_1\to\infty$) the error probability equals $$p_{\rm error}=\frac{e^{\alpha\beta_2}}{4\cosh{\alpha\beta_2}}=\frac{1}{2(1+e^{-2\alpha\beta_2})}\,,$$ hence, the error probability is decreasing with the temperature $T_2$ (increasing with the inverse temperature $\beta_2$) and the infinite temperature bath is best distinguishable from the ground state. Further, we see from Fig.~\ref{difalpha} that for small temperatures $T_1$ the error probability is gradually increasing until $T_2=T_1$, when it reaches the maximal value $P_{\rm error}=1/2$ (states becomes the same). After that the error probability decreases and converges (with $T_1\to\infty$) to a constant value $P_{\rm error}=\frac{1}{2}[1-\frac{1}{2}\tanh(\frac{\alpha}{T_2})]$. Depending on its value, $T_2$ is best distinguishable either from ground state ($T_1=0$), or the infinite temperature state ($T_1=\infty$). In fact, comparing the limiting values for small and large temperatures we may formulate the following observations.

\begin{itemize}
  \item{\emph{Critical temperature.}}
Consider a traceless qubit Hamiltonian $H$ with the operator norm
$||H||=\alpha$.  Define $T_*=\frac{\alpha}{\rm{arctanh}(\frac{1}{2})}$
and assume $T_2$ is fixed. If $T_2<T_*$, then the error probability achieves its minimum
for very high temperatures (infinity). If $T_2=T_*$, then the limiting error probabilities
  for large and small temperatures coincide. And if $T_2>T_*$, then the
  probability of error is minimal for low temperatures. The same observation was reported recently in \cite{CandeloroPRA2021}.

\item{\emph{Dependence on interaction strength.}}
Figure \ref{difalpha} illustrates how the error probability depends on the interaction strength $\alpha$. It turns out that for the smaller temperatures (differences) the smaller the interaction the smaller the error. However, for larger temperatures (differences of temperatures) the strength of the interaction might improve the discrimination of thermal states.

\item{\emph{Dependence on the difference of temperatures.}}
  Define the temperature difference $\Delta=T_2-T_1$. Is it easier to discriminate
  temperatures of a fixed $\Delta$ for smaller, or for larger temperatures? Figure \ref{delt} shows
  a general tendency that the difference for larger temperatures ($T>\alpha$) is more difficult
  to discriminate. In other words, the states of large temperatures are becoming less
  and less distinguishable. 
\end{itemize}
\subsection{Qudits: Finite dimensional linear harmonic oscillator}

  In this section we will address the case of $d$-dimensional systems (qudits) with the Hamiltonian 
\begin{eqnarray}
H=E_0\ket{0}\bra{0}+\cdots + E_{d-1}\ket{d-1}\bra{d-1}
\nonumber
\end{eqnarray}
As before, the thermal states $\varrho_\beta$ for $H$
and $H^\prime=H-E_0\pmb{I}$ coincide. Without loss of generality we will
assume the energy of the ground state is set to zero and 
label the excited energies by their difference from $E_0$, i.e.
$\alpha_j=E_j-E_0$ and
\begin{eqnarray}
  H^\prime=\alpha_1\ket{1}\bra{1}+\cdots \alpha_{d-1}\ket{d-1}\bra{d-1}\,.
\end{eqnarray}
Let us recall that the trace distance of the thermal states equals
\begin{equation}
  \nonumber
||\rho_1-\rho_2||= |\frac{1}{Z_1}-\frac{1}{Z_2}|+\sum_{j=1}^{d-1}|\frac{\exp(-\beta_1\alpha_j)}{Z_1}-\frac{\exp(-\beta_2\alpha_j)}{Z_2}|\,.
\end{equation}

Define the limiting values
  \begin{eqnarray}
    q_0=\lim_{T_1\to 0}||\varrho_1-\varrho_{2}||\,,\quad
    q_\infty=\lim_{T_1\to \infty}||\varrho_1-\varrho_{2}||
    \end{eqnarray}
being functions of $T_2$. If $q_\infty\geq q_0$, then $T_2$ is best discriminated
with temperatures $T_1$ from the area of very low temperatures.
Otherwise, $T_2$ is best discriminated with higher temperatures $T_1$.
In particular,
\begin{eqnarray}
q_0&=&|1-\frac{1}{Z_2}|+\frac{1}{Z_2}\sum_j e^{-\beta_2\alpha_j}=\frac{2(Z_2-1)}{Z_2}\,, \\ 
q_\infty&=&|\frac{1}{d}-\frac{1}{Z_2}|+\sum_j |\frac{1}{d}-\frac{\exp(-\beta_2\alpha_j)}{Z_2}|\,.
\end{eqnarray}
The identity $q_0=q_\infty$ identifies the critical value of temperature $T^*$ determining whether
the minimum of the error probability of distinguishing $T_2$ and $T_1$ is achieved for larger, or small
temperatures $T_1$ (see also the previous subsection).

  As an example consider the finite-dimensional analog of a linear harmonic oscillator.
  That is a system with the following Hamiltonian
\begin{equation}
  H= E_0 \pmb{I} + \alpha\sum_{j=1}^{d-1} j\ket{j}\bra{j} 
\end{equation}
Let us stress its energy levels are equidistant. Since in this case
$\frac{1}{d}-\frac{\exp(-\beta_2\alpha_1)}{Z_2} \geq 0$, it follows
\begin{eqnarray}
  q_\infty=\frac{2-Z_2}{Z_2}+\frac{d-2}{d}\,,
\end{eqnarray}
and subsequently the condition $q_0=q_\infty$ implies
\begin{equation}
\frac{d-1}{d+1} = \sum_{j=1}^{d-1}\exp{(-\frac{j\alpha}{T^*})} = \frac{\exp(-\frac{\alpha}{T^*})-\exp{(-\frac{N\alpha}{T^*})}}{1-\exp(-\frac{\alpha}{T^*})}\,.
\end{equation}
The Figure \ref{crit} illustrates the dependence of the critical temperatures $T^{*}$ on the system's dimension.
For the simplest cases we can evaluate the values also analytically 
\begin{eqnarray}
T^{*}&=&\frac{\alpha}{\ln(3)}  \enspace\enspace\enspace\enspace \text{for $d=2$} \nonumber \\
T^{*} &=&-\frac{\alpha}{\ln(\frac{2}{\sqrt{3}-1})} \enspace\enspace\enspace\enspace \text{for $d=3$}. 
\end{eqnarray}

\begin{figure}
\mbox{\includegraphics[scale=0.45]{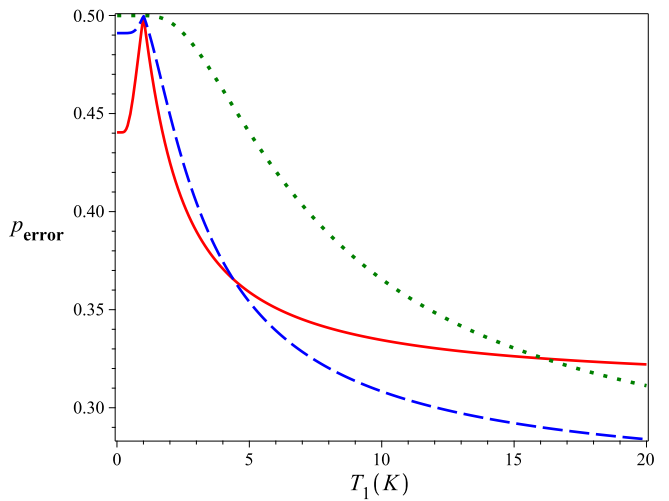}}
\caption{Probability of error for a fixed temperature $T_2 =1 (K)$ and some values of $\alpha$. $\alpha=1$(red line), $\alpha=2$ (dashed blue line) and $\alpha=5$ (dotted green line). It can be seen that for smaller temperatures, the smaller the interaction results in the smaller error.}
\label{difalpha}
\end{figure}

\begin{figure}
\mbox{\includegraphics[scale=0.45]{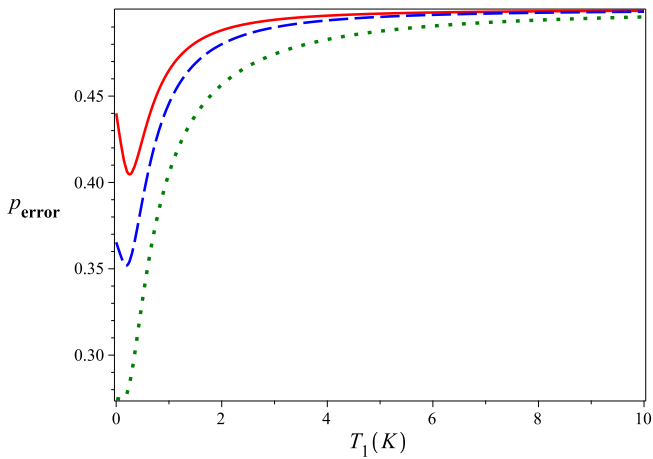}}
\caption{Probability of error for  fixed values of $\Delta T=0.5 (K)$ (red line),$\Delta T=1 (K)$ (dashed blue Line) and $\Delta T=5 (K)$ (dotted green Line). All plots are depicted for $\alpha=0.5$. 
  A general tendency can be seen that the difference for larger temperatures ($T>\alpha$) is more difficult
  to discriminate. }
\label{delt}
\end{figure}

\begin{figure} 
\mbox{\includegraphics[scale=0.45]{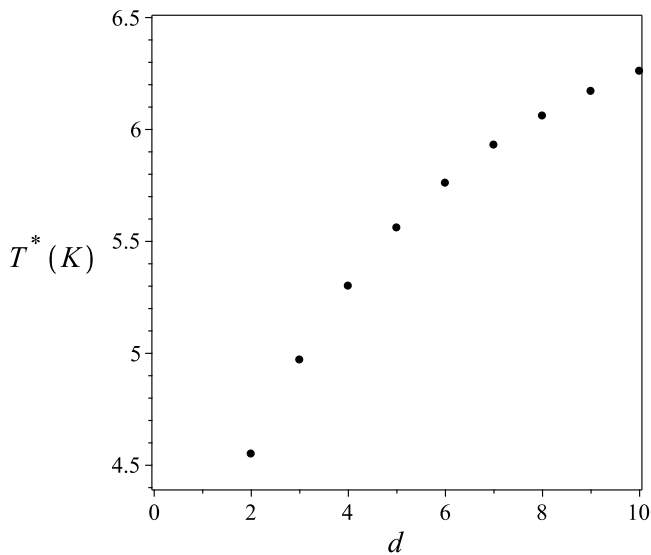}}
\caption{Critical Temperatures for different dimensions for fixed $\alpha=5$. The Critical Temperature increases for the systems with higher dimensions. }
\label{crit}
\end{figure}

\section{Case of multiple temperatures} \label{sec3}
In this section we will investigate the discrimination of $N$ thermal states
$\varrho_1,...,\varrho_N$ associated with inverse temperatures $\beta_1>\beta_2>\cdots >\beta_N$,
respectively. That is, the states are ordered accodring to their temperature in increasing order
and appear with a priori distribution $\eta_1,\dots,\eta_N$ ($\sum_j \eta_j=1$).
Let us introduce a success probability $p_{\rm success}(F)=\sum_j\eta_j\tr{\varrho_j F_j}$,
where the observation of outcome $F_j\geq O$ is used to conclude the state is $\varrho_j$.
It follows that the error probability for this measurement equals $p_{\rm error}(F)=1-p_{\rm success}(F)$,
thus, minimalization of the error probability over all positive operator-valued measure (POVMs) $F: F_1,\dots,F_N$ ($\sum_j F_j=I$)
is equivalent to the maximization of the success probability over $F$, i.e.
$p_{\rm error}=\min_F [p_{\rm error}(F)]=\min_F [1-p_{\rm success}(F)]=1-\max_F [p_{\rm success}(F)]$.
It is technically slightly simpler to optimize the success probability $p_{\rm success}$ and therefore
we will focus on this question.

It is known \cite{BaeNJP2013,BaePRA2013} that using the convex duality of the question of success optimization
$p_{\rm success}=\max_F p_{\rm success}(F)$ coincides with the following problem:
\begin{eqnarray}
  p_{\rm success}=\min_{K}\tr{K}\,, K\geq \eta_j\varrho_j \quad\forall j\,.
\end{eqnarray}
Let us use the notation $\varrho_j=\sum_l w_{jl}\Pi_l$, where $w_{jl}=\tr{\varrho_j\Pi_l}
=\frac{1}{Z_j}e^{-\beta_j E_l}$. Assuming the a priori distribution is not biased, i.e. $\eta_j=1/N$,
and writing $K=\frac{1}{N}\sum_l k_l\Pi_l$, it follows that the choice $k_l=\max_j w_{jl}$
guarantees the conditions $NK\geq \varrho_j$ hold for all $j$. The optimality of $K$ can
be argued \cite{BaeNJP2013} if we manage to design a measurement satisfying $\tr{F_j (K-\frac{1}{N}\varrho_j)}=0$.
Because of the positivity of  $F_j$ and $K-\frac{1}{N}\varrho_j$ the orthogonality under trace
is equivalent to the orthogonality of the supports of these operators. Since
$K-\frac{1}{N}\varrho_j=\frac{1}{N}(\sum_l (k_l - w_{jl}) \Pi_l)$, it follows that
$F_j =\sum_{l: \arg[\max_m w_{ml}]=j} \Pi_l$ are orthogonal to $K-\frac{1}{N}\varrho_j$ and
it is straightforward to verify that $\sum_j F_j=I$. In fact,
each of the projectors in the range of $F_j$ cancels out from the sum  $\sum_l (k_l - w_{jl}) \Pi_l$.
Therefore, the designed POVM is optimal \cite{BaeNJP2013}. In the case a state $\varrho_j$ does not maximize
any of the $\max_m w_{ml}$, then none of the term in the sum $\sum_l (k_l - w_{jl}) \Pi_l$ vanishes.
Consequently, the associated effect $F_j$ vanishes, thus, the defined optimal POVM does not lead
us to the conclusion for $\varrho_j$.

To summarize, we found that the effects
\begin{eqnarray}
  F_j=\sum_{l\in I_j}  \Pi_l\,,
\end{eqnarray}
where $I_j=\{l:\arg(\max_m\{\tr{\varrho_m\Pi_l}\}_m)=j\}$ and
$\Pi_l$ are eigenprojectors of system's Hamiltonian $H$, form
the optimal POVM minimizing the error probability for the discrimination
of thermal states $\varrho_1,\dots,\varrho_N$ with the uniform prior.
Whenever we observe the outcome $F_j$ we conclude the state is $\varrho_j$.
If for some $j$ the state $\varrho_j$ does not maximize $\tr{\varrho_j \Pi_l}$
for any of the projectors, then $F_j=O$ and the minimum-error procedure
never concludes such a state. As for the binary case also in this case the
optimal measurement can be implemented as an energy measurement. Observing the outcome
represented by the eigenprojector $\Pi_l$, i.e., measuring the value of energy $E_l$,
we conclude the state $\varrho_j$ maximizing the probability for this outcome.
The minimum error reads
\begin{eqnarray}
  \nonumber
  p_{\rm error}&=&1-\min_K \tr{K}\\
  \nonumber
  &=& 1-\frac{1}{N}\sum_l \tr{\Pi_l}\max\{\tr{\varrho_1\Pi_l},\dots,\tr{\varrho_N\Pi_l}\}\,.
\end{eqnarray}

Before we continue let us stress that the argumentation above applies also to the slightly more general case and we may formulate the following theorem.
\begin{theorem}
  Consider a set of mutually commuting states $\varrho_1,\dots,\varrho_N$ with apriori probabilities $\eta_1,\dots, \eta_N$. Let us denote by $\Pi_k$ the eigenprojectors of $\varrho_j$, thus, $\varrho_j=\sum_k w_{jk}\Pi_{k}$ for all $j$ and $w_{jk}=\tr{\varrho_j \Pi_k}$. Define the index set $I_j$ composed of indexes $k$ for which $\eta_j\tr{\varrho_j\Pi_k}$ is maximized by $\varrho_j$. Then the minimum-error discrimination measurement is composed of projectors $F_j=\sum_{k\in I_j} \Pi_k$ identifying the conclusion $\varrho_j$. The probability of success equals
\begin{eqnarray}
  \nonumber
  p_{\rm success}=\sum_k \tr{\Pi_k}\max\{\tr{\eta_1\varrho_1\Pi_k},\dots,\tr{\eta_N\varrho_N\Pi_k}\}\,.
\end{eqnarray}
\label{theorem}
  \end{theorem}
\begin{proof}
  It is known \cite{BaeNJP2013,BaePRA2013}
  that using the convex duality of the question of success
optimization $p_{\rm success}=\max_F p_{\rm success}(F)$ coincides with the
following problem
\begin{eqnarray}
  p_{\rm success}=\min_{K}\tr{K}\,, K\geq \eta_j\varrho_j \quad\forall j\,.
\end{eqnarray}
Define $x_k=\max\{\eta_1\tr{\varrho_1\Pi_k},\dots,\eta_N\tr{\varrho_N\Pi_k}\}$
and set $K=\sum_k x_k \Pi_k$. By construction $K\geq \eta_j\varrho_j$ for all
$j$, because $\varrho_j=\sum_k \tr{\varrho_j \Pi_k}\Pi_k$ 
guarantees this conditions $ K\geq \eta_j \varrho_j$ hold for all $j$. The optimality of $K$ can
be argued \cite{BaeNJP2013,BaePRA2013} if we manage to design a measurement satisfying $\tr{F_j (K-\eta_j\varrho_j)}=0$.
Because of positivity of the $F_j$ and $K-\eta_j\varrho_j$ the orthogonality under trace
is equivalent to the orthogonality of the supports of these operators. Since
$K-\eta_j\varrho_j=\sum_l (x_l -\eta_j w_{jl}) \Pi_l$, it follows that
$F_j =\sum_{l: \arg[\max_m w_{ml}]=j} \Pi_l$ are orthogonal to $K-\eta_j\varrho_j$ and
it is straightforward to verify that $\sum_j F_j=I$. In fact,
each of the projectors in the range of $F_j$ cancels out from the sum  $\sum_l (k_l - w_{jl}) \Pi_l$.
Therefore, the designed POVM is optimal. In the case a state $\varrho_j$ does not maximize
any of the $\max_m w_{ml}$, then none of the term in the sum $\sum_l (k_l - w_{jl}) \Pi_l$ vanishes.
Consequently, the associated effect $F_j$ vanish, thus, the defined optimal POVM does not result
in the conclusion for $\varrho_j$.
  \end{proof}

\subsection{Qubits}
The formula for error probability can be evaluated for the case of qubit thermal states
in the presence of a fixed Hamiltonian. Without loss of generality we will assume the Hamiltonian is traceless,
thus, $H=\alpha\vec{n}\cdot\vec{\sigma}$ for a unit vector $\vec{n}$. The energies equal
$E_\pm=\pm\alpha$ and thermal state
\begin{eqnarray}
  \varrho_j=\frac{e^{\alpha\beta_j}\Pi_0+e^{-\alpha\beta_j}\Pi_1}{2\cosh[\alpha\beta_j]}\,.
\end{eqnarray}
We showed previously (Sec. II.C) that
$\beta_j>\beta_k$ implies $\tr{\varrho_j\Pi_0}>\tr{\varrho_k\Pi_0}$
and $\tr{\varrho_j\Pi_1}<\tr{\varrho_k\Pi_1}$. Therefore, for decreasingly ordered inverse temperatures
it follows $\tr{\varrho_1\Pi_0}>\cdots >\tr{\varrho_N\Pi_0}$
and $\tr{\varrho_1\Pi_1}<\cdots <\tr{\varrho_N\Pi_1}$. Consequently, observations
of $\Pi_0,\Pi_1$ lead to conclusions $\varrho_1,\varrho_N$, respectively. In other words,
the discrimination of $N$ thermal states of qubits concludes with nonzero probability only
the states with minimal and maximal temperatures, i.e.
\begin{eqnarray}
\nonumber  p_{\rm error}&=&1-\frac{1}{N}(\tr{\varrho_1\Pi_0} +\tr{\varrho_N\Pi_1})\\
\nonumber &=& 1-\frac{1}{N}(\tr{(\varrho_1-\varrho_N)\Pi_0}+1)\\
  &=&\frac{N-1}{N}+\frac{1}{2N}\left(\frac{\sinh{\alpha(\beta_1-\beta_N)}}{\cosh{\alpha\beta_1}\cosh{\alpha\beta_N}} \right)
\end{eqnarray}

\subsection{Temperature threshold}
Consider now the following situation (analogous discrimination problem
  was studied in \cite{HerzogPRA2002}). As before, we are given a promise that the source is producing
one of increasingly ordered thermal states
$\varrho_1,\dots,\varrho_N$, however, our goal is to decide only whether the
temperature is above, or below some specified (threshold) value $T_c$. Such
a tempererature splits the states $\varrho_1,\dots,\varrho_N$ into
two subsets $S_{-}, S_{+}$ depending on whether their temperature
is smaller, or larger than $T_c$, respectively. Assuming all the
states are equally likely we may introduce density operators
\begin{eqnarray}
  \nonumber
  \varrho_{-}=\frac{1}{N_{-}}\sum_{j\in S_{-}} \varrho_j, \qquad
    \varrho_{+}=\frac{1}{N_{+}}\sum_{j\in S_{+}} \varrho_j\,,
\end{eqnarray}
where $N_-+N_+=N$ and $N_\pm$ labels the number of thermal states below and above the specified temperature, respectively. The discrimination problem is reduced to the discrimination of these averages of thermal states a priori distributed as $q_\pm=N_\pm/N$. The state $\varrho_\pm$ are themselves not thermal states (except for the qubit case), however, they are still commuting and diagonal in the energy basis. Therefore, Theorem \ref{theorem} applies and the postprocessed measurement in the energy basis is optimal. Whenever $q_-\tr{\varrho_- Q_j}>q_+\tr{\varrho_+ Q_j}$, the observation of the energy $E_j$ leads us to conclusion $\varrho_-$, i.e. the temperature is below $T_c$. In the case of the opposite inequality we conclude the temperature is above $T_c$. 

\begin{example}
Consider the case of qubit states. We know that each thermal qubit states
with $H=\alpha\sigma_z$ can be parametrized as follows
\begin{equation}
\varrho=\frac{1}{2} (I- \tanh(\frac{\alpha}{T})\sigma_z).
\end{equation}
Using this form the states $\varrho_{-}$ and $\varrho_{+}$ takes the form
\begin{eqnarray}
\varrho_{\pm}=\frac{1}{2} (I- \tanh(\frac{\alpha}{T_{\pm}})\sigma_z)\,, 
\end{eqnarray}
with temperatures
\begin{eqnarray}
T_{\pm}=\frac{1}{\alpha}\tanh^{-1}\left(\frac{\sum_{j\in S_\pm}\tanh(\frac{\alpha}{T_j})}{N_{\pm}}\right)\,.
\end{eqnarray}
The threshold problem reduces to the discrimination of states $\varrho_\pm$ with new prior probabilities
$q_{\pm}=N_{\pm}/N$. The optimal measurement (Theorem \ref{theorem}) consists of projectors onto eigenvectors
of the Hamiltonian
\begin{eqnarray}
\pi_1=\ket{1}\bra{1}\,, \nonumber\\
\pi_2=\ket{0}\bra{0}\,.
\end{eqnarray}
The conclusion made (above or below threshhold), when the outcome $\pi_j$ is registered, follows from the comparison of values $q_{\pm}tr(\varrho_\pm \pi_j)$.

In particular, consider the case of three states $\varrho_1,\varrho_2,\varrho_3$
with increasingly ordered temperatures $T_1=0<T_2<T_3$ and set the threshold temperature $0<T_c<T_2$. That is $\varrho_-=\ket{0}\bra{0}$ and
$\varrho_+=\frac{1}{2}(\varrho_2+\varrho_3)$ appearing with probabilities $q_-=1/3$ and $q_+=2/3$. Shifting the energy spectrum of the Hamiltonian
such that $E_0=0$ and $E_1=2\alpha$, we get $1\leq Z_j=1+e^{-2\alpha/T_j}\leq 2$.
The joint probability that the state $\varrho_\pm$ is measured
and the ground energy is observed equals
\begin{eqnarray}
\frac{1}{3}\tr{\varrho_-\pi_2}  &=&\frac{1}{3}\bra{0}(\ket{0}\bra{0})\ket{0}=\frac{1}{3}, \\
\frac{2}{3}\tr{\varrho_+\pi_2}&=&\frac{1}{3}\bra{0}(\varrho_2+\varrho_3)\ket{0}=\frac{1}{3} (\frac{1}{Z_2}+\frac{1}{Z_3}),
\end{eqnarray}
respectively. Since $1/2\leq 1/Z_j\leq 1$ it follows
$1/3(1/Z_2+1/Z_3)>1/3$, thus, the state $\varrho_+$ is concluded whatever
is the choice of temperatures $T_2,T_3$ . Let us stress that the registration
of the excited energy associated with the projectors $\pi_1$ also leads
to the same conclusion, because $\tr{\varrho_- \pi_1}=0$.
\end{example}

Let us recall that for the discrimination of a uniformly distributed pair of
qubit thermal states, the registration of the ground energy (associated with the
projector $\pi_2$) always leads to the conclusion that the thermal state
is the one with the smaller temperature.
whenever $T_1<T_2$. However, the situation is different
in the case of the discrimination of $\varrho_+$ and $\varrho_-$, thus, when
the threshold value is evaluated.
By construction $T_-<T_+$, however, it might happen that observation of the
ground energy implies ( to minimize the error probability) that the
temperature is above the threshold value, i.e. $T_+$ is concluded.

As a result, it is illustrated by the example that separating the qubit
ground state from the collection of excited ones is trivial in a sense
that both discrimination measurement outcomes lead to the same conclusion.
In other words, for the associated threshold discrimination problem
we necessarily conclude the temperature is above the threshold.

\section{Thermal States with different Hamiltonians and fixed temperature} \label{sec4}
Thermal states $\varrho_\beta(H)$ with the same Hamiltonian are mutually
commuting, thus, in a sense their discrimination is classical. Indeed,
we showed the optimal measurement is built out of the projective
measurement of the energy. In this section we will analyze the discrimination
of thermal states for different Hamiltonians.

Let us recall that for qubits any state can be understood
as a thermal state associated with some Hamiltonian. Therefore, the discrimination
of thermal states with different Hamiltonians define a completely general
discrimination problem. For the general dimension, the discrimination of thermal states
of arbitrary Hamiltonians is different from the discrimination of any collection
of states. However, it make sense to restrict ourselves to a specific case.
In a typical thermodynamical model a system is in contact with a thermal reservoir that
determines its temperature, thus, the temperature is known. However, the details of
the Hamiltonian might be uncertain and we are interested in resolving its identity.
In the context of discrimination this means to identify one of a finite numbers
of possibilities.

Imagine a situation of a spin system in a constant magnetic field characterized by its strength
$B$ and its direction $\vec{b}$. Consider a pair of thermal states of the same temperature
$\varrho_j=\varrho_{\beta}(H_j)$ with $H_j=B_j \vec{b}_j\cdot\vec{\sigma})$. The case of fixed direction, i.e.,
$\vec{b}_1=\vec{b}_2=\vec{b}$ reduces to discrimination of commuting density operators, thus,
in accordance with Theorem \ref{theorem} the optimal discrimination is achieved by the
projective measurement of the energy, i.e. in the eigenbasis of
$\vec{b}\cdot\vec{\sigma}$. Here we will address the question
when $B_1=B_2=B$, thus, the strength of the considered Hamiltonians is the same.
In such a case
\begin{equation}
p_{\rm error}=\frac{1}{2} (1-\frac{1}{\sqrt{2}} |\tanh\frac{B}{T})|\sqrt{1-\vec{b}_1\cdot\vec{b}_2})\,.
\end{equation}
For a fixed value of $B$ the error probability increases (Fig. \ref{penocomm}) with the
temperature as expected, because the thermal states with infinite temperature
is independent of the Hamiltonian. This increase with the temperature
may be compensated by increasing of the energy of the Hamiltonian. The error vanishes when
$\vec{b}_2=-\vec{b}_1$ and $T=0$, thus, the ground states are orthogonal. 

The considered thermal states take the form $\varrho_\beta(H_j)=\frac{1}{2}(I-\tanh\frac{B}{T}\vec{b}_j\cdot\vec{\sigma})$, thus,
in the Bloch sphere representation they are represented by Bloch vectors $\vec{v}_j=-\tanh\frac{B}{T}\vec{b}_j$. 
The associated optimal measurement is composed of projectors \cite{HaPRA2013}
\begin{eqnarray}
\pi_+=\frac{1}{2}[I+\frac{(\eta_1 \vec{v_1}-\eta_2\vec{v_2}).\sigma}{||\eta_1 \vec{v_1}-\eta_2\vec{v_2}||}], \nonumber\\
\pi_-=\frac{1}{2}[I-\frac{(\eta_1 \vec{v_1}-\eta_2\vec{v_2}).\sigma}{||\eta_1 \vec{v_1}-\eta_2\vec{v_2}||}]\,,
\label{pi1pi2}
\end{eqnarray}
where $\eta_j$ are a priori probabilities of the occurrence of states $\varrho_j$, respectively.
Expressing $\vec{v}_j$ via $\vec{b}_j$ and assuming, for simplicity, that $\eta_1=\eta_2=1/2$ it follows
$\pi_\pm = \frac{1}{2}(I\pm \frac{\vec{b}_1-\vec{b}_2}{||\vec{b}_1-\vec{b}_2||}\cdot\vec{\sigma})$.
As a result we see that optimal measurement is measuring the spin along the direction $\vec{b}_1-\vec{b}_{2}$.
Moreover, let us stress it is independent of the temperature $T$ and also of the strength of the magnetic
field $B$.

\begin{figure}
\mbox{\includegraphics[scale=0.45]{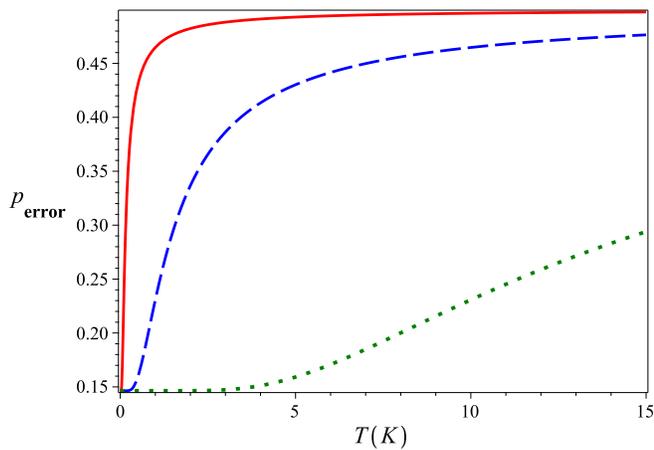}}
\caption{Probability of error in terms of $ T$ for $B=0.1 (T)$ (red Line) and $B=1 (T)$ (dashed blue Line) and $B=10 (T)$, (dotted green Line). For a fixed value of $B$ the error probability increases with the
temperature. It is because that the thermal states with infinite temperature
is independent of the Hamiltonian.}
\label{penocomm}
\end{figure}

\section{Conclusion} \label{conclusion}
In this paper, we studied the problem of minimum error discrimination of thermal states. The case of unambiguous discrimination simplifies as the only state that can be unambiguously concluded is the ground state associated with the zero temperature. In its full generality the investigated problem coincides with the general discrimination problem of a collection of pure and full-rank states, because the set of thermal states coincide with the set of such states. This means that the general problem is sufficiently complex to formulate general take-home messages, thus, we restricted ourselves to specific instances of the problem: fixed Hamiltonian, fixed temperature, and the so-called temperature threshold problem.

We first analyzed the discrimination of pairs of thermal states with the same Hamiltonian and showed that the measurement in the energy basis is the optimal one. Calculating the probability of error for the asymptotic cases $T_1=0$ and $T_1\to\infty$ we identified the critical temperature $T_*$. By comparing a temperature $T$ with $T_*$ we can find if it can be better discriminated with higher or lower temperature and if $T=T_*$ then the same result can be obtained with very high or very low temperatures. Then we analyzed the effect of the Hamiltonian strength $\alpha$ for the cases with some fixed temperature difference and observed that states of large temperatures are becoming less and less distinguishable. 

We generalized some of the results to the case of a $d$-dimensional Hamiltonian and in Theorem \ref{theorem} we formulated the optimal discrimination for a general set of mutually commuting states. Specially, the optimal discrimination of thermal states for qubits concludes only the largest and the smallest of the temperatures. Moreover, the observation of the ground or excited energy implied the temperature was the smallest or the largest, respectively. For general systems the minimum-error measurement identified at most $d$ different temperatures.

Further, we investigated the identification problem deciding whether the temperature of thermal states was above, or below some threshold value of the temperature $T_c$. We pointed out the situations when this problem became trivial in a sense that only the ``above" conclusion was possible, hence, no discrimination measurement was needed. Finally, we extended the investigation to a noncommutative case. We showed that when the qubit Hamiltonians were different (unitarily related), but the temperature was fixed, then the optimal discrimination measurement was independent of the (fixed) temperature and the interaction strength. Naturally, this question is an instance of the discrimination of Hamiltonians, thus, the optimal measurement is qualitatively related to discrimination of the associated energy measurements and the best distinguishable situation coincide with the discrimination of ``orthogonal" Hamiltonians $\pm H$. However, a deeper analysis is required to relate the ``thermal" discrimination of Hamiltonians with the discrimination of the associated observables. 

\acknowledgments{This research was supported by projects OPTIQUTE (APVV-18-0518) and HOQIT (VEGA 2/0161/19/).}


\end{document}